\documentclass[sigconf,nonacm]{acmart}

\usepackage{listings}
\usepackage{algorithm}
\usepackage[noend]{algpseudocode}
\usepackage{tikz}
\usetikzlibrary{positioning,arrows.meta,fit,shapes.geometric,backgrounds,decorations.pathreplacing}
\usepackage{booktabs}

\AtEndPreamble{%
  \theoremstyle{definition}%
}

\lstdefinelanguage{3D}{
  morekeywords={entrypoint,typedef,struct,UINT8,UINT16,UINT32,BOOLEAN},
  sensitive=true,
  morecomment=[l]{//},
  morecomment=[s]{/*}{*/},
  morestring=[b]",
}

\newcommand{\TRUE}{\textsc{True}}
\newcommand{\FALSE}{\textsc{False}}

\setcopyright{none}
\title{Access Control as Verified Parse Constraints}

\author{Saranachon Iammongkol}
\email{iamfo470@student.otago.ac.nz}
\affiliation{%
  \institution{University of Otago}
  \department{School of Computing}
  \city{Dunedin}
  \country{New Zealand}
}

\author{Zhiyi Huang}
\email{zhiyi.huang@otago.ac.nz}
\affiliation{%
  \institution{University of Otago}
  \department{School of Computing}
  \city{Dunedin}
  \country{New Zealand}
}

\author{David Eyers}
\email{david.eyers@otago.ac.nz}
\affiliation{%
  \institution{University of Otago}
  \department{School of Computing}
  \city{Dunedin}
  \country{New Zealand}
}

\begin{abstract}
Commercial security gateways repeatedly ship implementation bugs in
the code path between the network and the policy decision:
hand-written enforcement logic that diverges from the policy author's
intent, and ad-hoc request parsers at the network boundary that
introduce memory-safety flaws of their own. In both cases the bug is
in the deployed enforcement code, not in the policy.
Existing approaches either leave the enforcement runtime unverified or
connect a formal model to a hand-written engine only by differential
testing~\cite{cutler_2024_cedarnewlanguage}.

Our contribution is a class result: a forward-only, backtrack-free
EverParse validator is a verified recognizer for a bounded,
finite-state class, and access-control decision functions with
fixed-offset fields and bounded disjunction belong to it, so one
machine-checked proof transfers to every policy in the class rather
than being re-established per policy. Concretely, we encode a bounded
policy language's decision function into a fixed-size byte buffer and
verify the enforcement code once---covering all byte values---with an
SMT solver, proving the validator accepts if and only if the decision
function accepts, for every policy, request, and session. Editing rule
content over a fixed endpoint set then needs no new proof; adding
endpoints reruns the toolchain; extending the language needs new
proofs. We establish faithful enforcement of a policy, not that a
policy is itself secure.

The verified gate is platform-independent, requiring only EverParse/Z3
and a C compiler, whose correctness we assume. We demonstrate a
deployment on the seL4 microkernel, which ensures every request passes
through the gate and that unverified components cannot corrupt the
verified enforcement chain.
\end{abstract}


\begin{document}
\maketitle

\section{Introduction}
\label{sec:introduction}

Commercial security gateways have repeatedly shipped enforcement
bugs that share a structural pattern across vendors and years.
CVE-2024-0012~\cite{nistnationalvulnerabilitydatabase_2024_cve20240012panosauthentication}
(Palo Alto PAN-OS, unauthenticated administrative access because the
management web interface trusted an attacker-settable HTTP header as
authentication state),
CVE-2025-0108~\cite{nistnationalvulnerabilitydatabase_2025_cve20250108panosauthentication}
(Palo Alto PAN-OS, authentication bypass in the same interface because
two reverse-proxy layers disagreed on URL path normalisation),
CVE-2025-20362~\cite{nistnationalvulnerabilitydatabase_2025_cve202520362authorizationbypass}
(Cisco Secure Firewall ASA/FTD, unauthenticated access to restricted
VPN endpoints via a URL path-normalisation flaw, identified by Rapid7
as a patch bypass of CVE-2018-0296 in the same WebVPN component seven
years earlier), and
CVE-2019-11816~\cite{nistnationalvulnerabilitydatabase_2019_cve201911816incorrectaccess}
(OPNsense, privilege escalation via a flawed URL comparison) were not
caused by misconfigured policies. In each case the code path between
the network and the policy decision---whether a wrong comparator, a
trusted header, or two components disagreeing on the shape of a
URL---contained an implementation bug that caused the decision to
silently deviate from the intended policy, and patching one incident
did not prevent the next error of the same kind in the same class of
code.

The recurrence is structural. Access control enforcement is
almost always hand-written imperative code that parses
variable-length inputs, canonicalises fields, and applies rule
logic at several interfaces. Every boundary where strings, URLs,
or headers pass between components is a candidate for the
enforcement mechanism to diverge from the policy author's intent.
Section~\ref{sec:opnsense-coverage} analyses the OPNsense
access-control surface and maps every vulnerability class our
approach eliminates against its actual enforcement code.

We ask whether the translation from policy specification to deployed
enforcement code can be made correct by construction, at the cost of
a restricted policy language. The observation that makes this possible
is a class result: a forward-only, backtrack-free EverParse validator
is a verified recognizer for a bounded, finite-state class, and a
useful fragment of access control---role comparison, capability
bitfields, exact-match resources---belongs to it, because its decision
function is built from fixed-offset fields and bounded disjunction with
no backtracking. The single correctness proof therefore transfers to
every policy the fragment can express.
EverParse~\cite{ramananandro_2019_everparseverifiedsecure}, a
verified parser generator, accepts exactly this shape, and compiles
a 3D specification---EverParse's specification
language~\cite{swamy_2022_hardeningattacksurfaces}---directly into
C through a trusted toolchain, so the deployed enforcement code and the specification
are not separate artefacts.

The price is forward-only parsing: EverParse does not backtrack, which
rules out wildcard path matching and other patterns needing
speculative, variable-length interpretation at runtime---a limit of
the recognizer model, not of termination. Within that constraint the
language expresses primitives from NIST
RBAC~\cite{ferraiolo_2001_proposedniststandard}---role comparison,
permission checking via capability bitfields, and resource
matching---over a fixed number of rule slots. A single verified binary
then enforces every policy the language can express over a fixed
endpoint set: editing rule content is a runtime data change, while
adding endpoints or extending the language reruns the toolchain.
This leverage matters at scale: hand-proved verified C is
expensive---seL4's 10{,}000-line kernel rests on ${\approx}480{,}000$
lines of Isabelle/HOL
proof~\cite{klein_2014_comprehensiveformalverification}---whereas our
245-line 3D specification compiles to ${\approx}4{,}000$ lines of
verified C, Z3 discharging the obligations automatically.
This work is inspired by
Cedar~\cite{cutler_2024_cedarnewlanguage}, which links a Lean model to
a hand-written Rust engine by differential testing; we instead generate
the deployed C from the specification, so there is no engine to test
and deny-by-default is structural (\S\ref{sec:related-work}).

\paragraph{Contributions.}
\begin{enumerate}
  \item \textbf{Access-control enforcement as membership in a verified
        recognizer's class.}
        A bounded access-control fragment---fixed-offset fields,
        bounded disjunction, no backtracking---lies inside the class a
        forward-only EverParse validator provably recognizes, so one
        proof covers every policy in it. Consequently
        (a)~the representation-confusion class of bugs behind the
        gateway CVEs in \S\ref{sec:introduction} is excluded by
        construction, and (b)~deny-by-default on absent rules is a
        property of the encoding, not a runtime discipline.

  \item \textbf{A three-tier update model, and a verified path from
        HTTP to decision.}
        Because the proof is universal over the fragment, editing rule
        content over a fixed endpoint set needs no new proof, adding
        endpoints reruns the toolchain, and extending the language
        needs new proofs; a verified extraction layer carries requests
        from raw HTTP to the gate, closing the gap between the formal
        decision function and the deployed enforcement code.

  \item \textbf{A reference deployment on a verified OS that covers a
        real firewall.}
        We integrate the gate on seL4 inside a component-isolated
        architecture and demonstrate that the resulting language
        covers the complete OPNsense~19.1.7 access-control list at
        microsecond decision latency.
\end{enumerate}

\section{Overview}
\label{sec:overview}

The gateway CVEs in \S\ref{sec:introduction} share a structural
shape: each is a disagreement about how an input is read---a
header trusted when it should not have been, a URL interpreted
differently by two components, a comparator matching the wrong
substring. These are errors of reading, not errors of deciding.
If the recurring failure is in how requests are read, the right
discipline is one designed for correct reading: data format
validation. Under this reframing the access-control question
becomes ``is this byte buffer well-formed and authorised?'' The
rest of this section traces a concrete request through the
system to show how this reframing works in practice.

Figure~\ref{fig:core-idea} gives the picture: a verified
preprocessor bounds the unbounded HTTP request into a fixed-size
metadata block, and the access-control decision is the output of
a forward-only verified validator scanning this block.

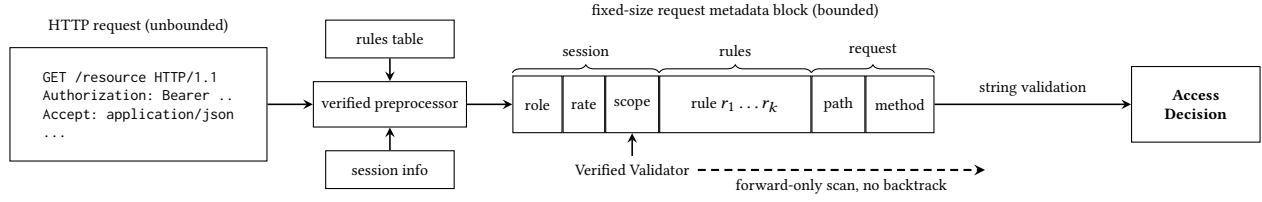
\begin{figure*}[t]
\centering
\begin{tikzpicture}[
  box/.style={rectangle, draw, font=\scriptsize\ttfamily,
              inner sep=3pt, align=left},
  procbox/.style={rectangle, draw, font=\scriptsize, align=center,
                  inner sep=3pt, minimum width=1.7cm, minimum height=0.6cm},
  cell/.style={rectangle, draw, minimum height=0.75cm,
               font=\scriptsize, align=center, inner sep=2pt},
  lbl/.style={font=\scriptsize, align=center},
  arrow/.style={->, >=stealth, semithick}
]

\node[box, minimum width=3.4cm, minimum height=1.5cm]
  (http) at (0,0)
  {GET /resource HTTP/1.1\\Authorization: Bearer ..\\Accept: application/json\\...};
\node[lbl, above=2pt of http.north] {HTTP request (unbounded)};

\node[procbox, right=0.6cm of http] (proc) {verified preprocessor};

\node[procbox, minimum height=0.5cm, above=0.3cm of proc] (rules) {rules table};
\draw[arrow] (rules.south) -- (proc.north);

\node[procbox, minimum height=0.5cm, below=0.3cm of proc] (session) {session info};
\draw[arrow] (session.north) -- (proc.south);

\draw[arrow] (http.east) -- (proc.west);

\node[cell, minimum width=0.65cm, right=0.6cm of proc] (c0) {role};
\node[cell, minimum width=0.55cm, right=0pt of c0]     (c1) {rate};
\node[cell, minimum width=0.7cm, right=0pt of c1]      (c2) {scope};
\node[cell, minimum width=2.0cm, right=0pt of c2]      (c3) {rule $r_1 \ldots r_k$};
\node[cell, minimum width=0.7cm, right=0pt of c3]      (c4) {path};
\node[cell, minimum width=0.9cm, right=0pt of c4]      (c5) {method};

\draw[arrow] (proc.east) -- (c0.west);

\draw[decorate, decoration={brace, amplitude=3pt, raise=1pt}]
  (c0.north west) -- (c2.north east)
  node[midway, above=4pt, font=\scriptsize] {session};
\draw[decorate, decoration={brace, amplitude=3pt, raise=1pt}]
  (c3.north west) -- (c3.north east)
  node[midway, above=4pt, font=\scriptsize] {rules};
\draw[decorate, decoration={brace, amplitude=3pt, raise=1pt}]
  (c4.north west) -- (c5.north east)
  node[midway, above=4pt, font=\scriptsize] {request};

\node[lbl, above=18pt of c3.north] {fixed-size request metadata block (bounded)};

\node[lbl, below=0.3cm of c2.south] (head) {Verified Validator};
\draw[arrow] (head.north) -- (c2.south);
\draw[->, >=stealth, thick, densely dashed]
  (head.east) -- ++(3.8cm, 0)
  node[midway, below, font=\scriptsize] {forward-only scan, no backtrack};

\node[cell, minimum width=1.7cm, minimum height=1.0cm, right=2.6cm of c5,
      font=\scriptsize\bfseries, align=center]
  (verdict) {Access\\Decision};

\draw[arrow] (c5.east) -- (verdict.west)
  node[midway, above, font=\scriptsize] {string validation};

\end{tikzpicture}
\caption{Gate architecture.}
\label{fig:core-idea}
\end{figure*}

Suppose a small deployment protects the endpoint
\texttt{/api/policy}, mapped at compile time to the 32-bit path
identifier \texttt{0x44444444}. The policy author writes two
rules into the gate's rule table:
\begin{itemize}
  \item Rule slot 1: path \texttt{0x44444444}, method mask
        \texttt{0x03} (\texttt{GET}\,+\,\texttt{POST}), minimum
        role \textsc{admin} (2), required scope \texttt{0x0004}
        (\textsc{configure}).
  \item Rule slot 2: path \texttt{0x44444444}, method mask
        \texttt{0x01} (\texttt{GET} only), minimum role
        \textsc{operator} (1), required scope \texttt{0x0001}
        (\textsc{read\_sensors}).
\end{itemize}
The remaining slots (3--8) are filled with the inactive sentinel
$\bot = \texttt{0xDEADDEAD}$; their role will become clear
shortly.

When a request arrives, the verified extractor writes the
session's authenticated role, rate counter, and scope into fixed
byte positions, appends the rule slots from the current
in-memory policy, and appends the request's path identifier and
method byte---producing a single flat buffer whose layout is
formalised in Table~\ref{tab:buffer-layout}
(\S\ref{sec:verified-mechanism}). The validator then scans this
buffer forward-only: at each rule slot, it checks whether the
request's path and method bytes match the slot's bytes at the
corresponding offsets, whether the session's role byte is at
least the slot's minimum, and whether the required scope bits
are present in the session's scope. If any slot matches, parsing
succeeds and the request is accepted; otherwise it is rejected.

An admin (role \texttt{0x02}, scope \texttt{0x0007}) sending
\texttt{POST /api/policy} matches Rule~1: the paths are equal,
\texttt{POST} is within the method mask, the role meets the minimum
($2 \geq 2$), and the required scope bit is present---accepted. An
operator (role \texttt{0x01}, scope \texttt{0x0003}) sending the same
request is denied: Rule~1 fails on role ($1 < 2$) and Rule~2 on method
(its mask allows only \texttt{GET}). The same operator sending
\texttt{GET /api/policy} is accepted by Rule~2. The binary never
changed---only the bytes differ, and the same verified comparisons
yield different outcomes.

Suppose auditors should later read \texttt{/api/policy} with no scope
requirement. The administrator writes a third rule into slot~3 at
runtime---path \texttt{0x44444444}, method \texttt{GET}, minimum role
\textsc{operator}, scope \texttt{0x0000}---changing the slot-3 bytes
from the $\bot$ sentinel to the rule's values. This is a data edit, not
a code change: the binary is exactly what Z3 verified, and the
soundness theorem (\S\ref{sec:soundness}) covers \emph{all} byte values
at every offset, so the new rule holds from the moment its bytes are
written. Editing rule content over the compiled endpoint set is thus
already covered; adding endpoints or growing $K$ reruns the toolchain,
and extending the language needs new proofs.

The sentinel $\bot$ that filled slot~3 deserves a closer look. Two
reserved values sit outside the author's identifier space
(Assumption~A1, \S\ref{sec:domains}): path ID zero, which the extractor
rejects with HTTP~404 before the gate, and
$\bot = \texttt{0xDEADDEAD}$, which fills inactive slots. A slot holding
$\bot$ fails the path comparison against every request that reaches the
gate---structurally inert---so if \emph{all} eight slots hold $\bot$,
every request is denied. Deny-by-default is thus a consequence of the
encoding, not a runtime check a developer could forget.

We target network attackers who control all HTTP request fields;
soundness bounds accept/reject for every constructible request.
Attacker-controlled HTTP fields cannot supply identity: role and scope
enter the buffer only through the F*-verified extractor, never from the
request itself. Out of scope: compromised credentials at the target
privilege, toolchain bugs, and post-gate application behaviour.

Nine core assumptions (toolchain, byte-order, and the byte-copy glue)
are platform-independent; three more cover the seL4 isolation substrate
and can be replaced by equivalent guarantees off-seL4
(Table~\ref{tab:assumptions}). The rest of the paper formalises this
picture---the policy language, the encoding and its EverParse/Z3
verification, and soundness (\S\ref{sec:verified-mechanism}); the seL4
deployment (\S\ref{sec:reference-deployment}); evaluation
(\S\ref{sec:evaluation}); and related work
(\S\ref{sec:related-work}).

\section{Verified Enforcement Mechanism}
\label{sec:verified-mechanism}

\subsection{Policy Language}
\label{sec:policy-language}

\subsubsection{Domains}
\label{sec:domains}

The gate operates on fixed-size security metadata produced by the
F*-verified extractor (\S\ref{sec:reference-deployment}): a resource
identifier, a method, a role level, a scope bitfield, and a rate
counter. The resource identifier is a 32-bit integer representing a
known URL path, assigned by the policy author at compile time; the
extractor maps incoming paths to identifiers by exact string
comparison and rejects unknown paths with HTTP~404 before the gate
is reached. Two reserved values are excluded from the author's
identifier space (Assumption~A1): zero marks an unrecognised
request at the extractor layer and never reaches the gate; the
sentinel $\bot = \texttt{0xDEADDEAD}$ (\texttt{INACTIVE\_PATH\_ID})
fills unused rule slots so inactive rules cannot match any valid
request. Wherever $\bot$ appears below it refers to this rule-slot
sentinel, not to the extractor zero value.
A 32-bit identifier space accommodates real-world
deployments---the OPNsense firewall ACL we evaluate in
\S\ref{sec:evaluation} defines 420 distinct paths.

All domains are configurable at compile time.
Table~\ref{tab:domains} summarises their general form;
Table~\ref{tab:demo-config} lists the concrete values used in our
HTTP demonstration deployment.
The same model accommodates other request-response protocols
(e.g.\ the Constrained Application Protocol (CoAP) or the MODBUS
industrial fieldbus) by reassigning the method and resource domains.

\begin{table}[t]
\centering\small
\caption{Policy language domains (general form). All domains are
         configurable at compile time.}
\label{tab:domains}
\begin{tabular}{@{}p{0.8cm}p{2.4cm}p{4.2cm}@{}}
\toprule
Symbol & Domain & Description \\
\midrule
$\mathbf{P}$
  & 32-bit identifiers\newline
    (0 and \texttt{0xDEADDEAD}\newline prohibited)
  & Resource identifiers, fixed at\newline compile time; reserved values\newline
    serve as rejection and inactive-rule\newline sentinels respectively \\[4pt]
$\mathbf{M}$
  & $\{0, \ldots, \texttt{0x3F}\}$
  & 6-bit operation bitmask;\newline
    each bit is a protocol operation;\newline
    a rule's mask may combine operations \\[4pt]
$\mathbf{R}$
  & $\{0, \ldots, R_{\max}\}$
  & Role levels, totally ordered \\[4pt]
$\mathbf{N}$
  & $\{0, \ldots, N_{\max}\!-\!1\}$
  & Rate counter;\newline
    enables inequality constraints,\newline
    e.g.\ $c < N_{\max}$ \\[4pt]
$\Sigma$
  & $\{0, \ldots, 2^{16}\!-\!1\}$
  & Scope bitfield (16-bit);\newline
    each bit is an independent capability;\newline
    distinguishes permissions within the\newline
    same role \\
\bottomrule
\end{tabular}
\end{table}

\begin{table}[t]
\centering\small
\caption{Demonstration deployment configuration (HTTP).
         Values defined in \texttt{path\_ids.h} and
         \texttt{HTTP.Extract.Types.fst}.}
\label{tab:demo-config}
\begin{tabular}{@{}p{1.6cm}p{5.8cm}@{}}
\toprule
Domain & Configuration \\
\midrule
Paths ($\mathbf{P}$)
  & \texttt{/api/login}$\to$\texttt{0x11111111},\newline
    \texttt{/api/logout}$\to$\texttt{0x22222222},\newline
    \texttt{/api/status}$\to$\texttt{0x33333333},\newline
    \texttt{/api/policy}$\to$\texttt{0x44444444} \\[4pt]
Methods ($\mathbf{M}$)
  & \texttt{GET}=\texttt{0x01},
    \texttt{POST}=\texttt{0x02},\newline
    \texttt{PUT}=\texttt{0x04},
    \texttt{DELETE}=\texttt{0x08},\newline
    \texttt{PATCH}=\texttt{0x10},
    \texttt{HEAD}=\texttt{0x20} \\[4pt]
Roles ($\mathbf{R}$)
  & \texttt{ROLE\_NONE}=0,
    \texttt{ROLE\_OPERATOR}=1,\newline
    \texttt{ROLE\_ADMIN}=2
    \quad($R_{\max}=2$) \\[4pt]
Rate ($\mathbf{N}$)
  & $N_{\max} = 50$ \\[4pt]
Scope ($\Sigma$)
  & bit\,0: \texttt{SCOPE\_READ\_SENSORS},\newline
    bit\,1: \texttt{SCOPE\_WRITE\_SENSORS},\newline
    bit\,2: \texttt{SCOPE\_CONFIGURE}\newline
    (bits 3--15 reserved) \\
\bottomrule
\end{tabular}
\end{table}

A request carries exactly one method bit; unknown or malformed
operations map to zero and cannot match any rule.

\subsubsection{Request and Session State}
\label{sec:request-session}

The gate evaluates two inputs at decision time: the \emph{request}
describing what is being asked, and the \emph{session state} describing
who is asking.

\begin{definition}[Request]
\label{def:request}
A request is a pair
\[
  q \;=\; (p_{\mathrm{req}},\; m_{\mathrm{req}}) \;\in\; \mathbf{P} \times \mathbf{M}
\]
where $p_{\mathrm{req}} \in \mathbf{P}$ is the \textbf{requested path identifier} and
$m_{\mathrm{req}} \in \mathbf{M}$ is the \textbf{requested method}.
\end{definition}

\begin{definition}[Session state]
\label{def:session-state}
A session state is a triple
\[
  s \;=\; (\alpha,\; c,\; \sigma)
  \;\in\; \mathbf{R} \times \mathbf{N} \times \Sigma
\]
where:
\begin{itemize}
  \item $\alpha \in \mathbf{R}$ is the \textbf{authenticated role}---the
        trust level assigned to the user at login
  \item $c \in \mathbf{N}$ is the \textbf{rate counter}---the number of
        requests made by this session in the current window
  \item $\sigma \in \Sigma$ is the \textbf{authenticated scope}---the
        capability bits granted to this session at authentication time
\end{itemize}
\end{definition}

The role is totally ordered, so rules express minimum-level checks
($\alpha \geq \rho_i$, where $\rho_i$ is the minimum role required by
rule~$r_i$, defined below).
The scope bitfield provides fine-grained capability delegation within
a role level: sessions at the same role may hold different capabilities.

Each session carries a single active role determined at authentication
time, with a per-session scope bitfield for capability-level
differentiation that a role ordering alone does not capture.
The language captures a subset of the operations found in role-based
access control systems~\cite{ferraiolo_2001_proposedniststandard}:
role comparison (as a total order rather than the partial order of
NIST RBAC), permission checking (via scope bitfields), and resource
matching. It does not model role hierarchies (lattices), multi-role
sessions, or separation of duty constraints; extending to these
features is discussed in \S\ref{sec:conclusion}.
This two-dimensional design enables the fixed-size encoding required
for universal verification (\S\ref{sec:encoding}).

Session state is treated as trusted input; its correctness---that the
right role and scope are assigned at login---is assumed.
The seL4 deployment (\S\ref{sec:architecture}) demonstrates one
practical approach: the authenticator is unverified but isolated, so a
bug may assign incorrect values but cannot corrupt the policy gate.

\subsubsection{Access Control Rules and Policy}
\label{sec:rules}

An access control rule specifies the conditions under which a request is
permitted.

\begin{definition}[Rule]
\label{def:rule}
An access control rule is a 4-tuple
\[
  r \;=\; (p,\; m,\; \rho,\; \sigma)
  \;\in\; \mathbf{P} \times \mathbf{M} \times \mathbf{R} \times \Sigma
\]
where $p$ is the \textbf{path identifier}, $m$ is the \textbf{method
mask}, $\rho$ is the \textbf{minimum role}, and $\sigma$ is the
\textbf{required scope}.
How each field is matched against the request and session is defined in
\S\ref{sec:decision} (Table~\ref{tab:rule-layout}).
\end{definition}

A policy is a fixed-size tuple of rules.
The fixed-size constraint is a design requirement of the verification
framework: the 3D specification defines a statically known buffer
structure.
Variable numbers of active rules are accommodated by padding unused
slots with \emph{inactive} rules whose path is set to the
$\bot$ sentinel (\S\ref{sec:domains}).

\begin{definition}[Policy]
\label{def:policy}
A policy is a tuple $\pi = (r_1, r_2, \ldots, r_K)$ of exactly $K$ rules,
where $K$ is fixed at compile time ($K = 8$ in our implementation).
A rule $r_i$ is \emph{active} if $p_i \neq \bot$, and \emph{inactive}
if $p_i = \bot$.
The F*-verified extractor guarantees $p_{\mathrm{req}} \neq \bot$
for every request, so inactive rules can never match.
\end{definition}

\subsubsection{Policy Decision Function}
\label{sec:decision}

Given a policy $\pi$ (Definition~\ref{def:policy}), a request $q$
(Definition~\ref{def:request}), and a session state $s$
(Definition~\ref{def:session-state}), the gate must decide whether to
accept or deny the request.
We write $a \subseteq b$ for \emph{bitwise subset}, i.e.,
$b \mathbin{\&} a = a$ (every bit set in $a$ is also set in $b$).

\begin{definition}[Rate counter update]
\label{def:state-transition}
On each authenticated request, the session counter is updated
unconditionally before the policy decision:
\begin{equation}
c' \;=\; \begin{cases}
  c + 1 & \text{if } \Delta t < T_{\mathrm{reset}} \\
  1     & \text{otherwise (window expired, counter resets)}
\end{cases}
\end{equation}
In the demonstration deployment, $T_{\mathrm{reset}} = 60$\,s.
The rate constraint reflects standard API practice for mitigating
brute-force and resource exhaustion.
\end{definition}

\begin{definition}[Policy decision]
\label{def:policy-decision}
After the counter update (Definition~\ref{def:state-transition}),
a request is accepted if and only if the rate constraint holds
over the updated counter $c'$ and at least one rule matches.
The per-rule match function over the request $q$, session $s$,
and rule $r_i$ is:
\begin{equation}
\label{eq:match}
\begin{split}
f(q,\; s,\; r_i) \;\triangleq\;\;
  &\underbrace{(p_{\mathrm{req}} \;=\; p_i)}_{\text{path}} \;\wedge\;
   \underbrace{(m_{\mathrm{req}} \;\subseteq\; m_i)}_{\text{method}} \\
  \wedge\;\;
  &\underbrace{(\alpha \;\geq\; \rho_i)}_{\text{role}} \;\wedge\;
   \underbrace{(\sigma_i \;\subseteq\; \sigma)}_{\text{scope}}
\end{split}
\end{equation}
The request is accepted when:
\begin{equation}
\label{eq:decision}
D_{\mathrm{accept}}(\pi,\; q,\; s) \;\triangleq\;
  \underbrace{(c' < N_{\max})}_{\text{rate (universal)}} \;\wedge\;
  \underbrace{\bigvee_{i=1}^{K}
    f(q,\; s,\; r_i)}_{\text{access (existential)}}
\end{equation}
\end{definition}

\begin{table}[t]
\centering\small
\caption{Policy decision structure. The rate check is universal
         (applied once); the remaining checks are per-rule
         (any matching rule $r_i$ accepts).}
\label{tab:rule-layout}
\begin{tabular}{@{}lccc@{}}
\toprule
Check & Request/Session & & Rule/Bound \\
\midrule
Rate   & $c$                  & $<$        & $N_{\max}$ \\
Path   & $p_{\mathrm{req}}$   & $=$        & $p_i$ \\
Method & $m_{\mathrm{req}}$   & $\subseteq$ & $m_i$ \\
Role   & $\alpha$             & $\geq$     & $\rho_i$ \\
Scope  & $\sigma$             & $\subseteq$ & $\sigma_i$ \\
\bottomrule
\end{tabular}
\end{table}

The counter is incremented on every authenticated request, regardless
of whether the request is ultimately accepted or denied; denied
requests still count toward the rate limit.
The reset window $T_{\mathrm{reset}}$ (a time, $60$\,s here) and the
count threshold $N_{\max}$ are independent quantities in different
units, with no ordering between them.

The decision function $D_{\mathrm{accept}}$ is defined independently of any particular
solver or encoding.
The encoding of $D_{\mathrm{accept}}$ as verified parse constraints is addressed in
\S\ref{sec:encoding}.

Section~\ref{sec:soundness} proves that the encoding and implementation
preserve the decision: the verified C code accepts exactly the requests
that $D_{\mathrm{accept}}(\pi, q, s)$ admits.
\subsection{Encoding and Verification}
\label{sec:encoding}

\subsubsection{Definition}

We define an encoding function that translates the policy decision
inputs into a flat byte buffer matching the
\texttt{\_Access\-Request} struct in the EverParse 3D specification
(\texttt{Rbac\-Policy.3d}, 245~lines; included in the artefact).

\begin{definition}[Encoding function]
\label{def:encoding}
The encoding function for the AccessRequest validator packs the policy
$\pi$ (Definition~\ref{def:policy}), request $q$
(Definition~\ref{def:request}), and session state $s$
(Definition~\ref{def:session-state}) into a single buffer:
\begin{equation}
E_{\mathrm{access}}(\pi, q, s)
  \;\longrightarrow\; \mathbb{B}^{75}
\end{equation}
The 75-byte buffer is structured as shown in Table~\ref{tab:buffer-layout}.
Each of the $K = 8$ rule slots occupies 8 bytes:
4 bytes for the path identifier
(little-endian), 1 byte for the HTTP method, 1 byte for the minimum role, and
2 bytes for the required scope (little-endian).
Unused rule slots are filled with inactive rules ($p_i = \bot$).
The two trailing bytes (offsets 73--74) are \emph{phantom bytes}---fields
whose value is irrelevant but whose EverParse where-clause (defined in
\S\ref{sec:encoding-structure}) anchors a logical constraint.
\end{definition}

The buffer is constructed entirely within the policy gate's address
space and is never transmitted over the network (PolicyGate is the
component that hosts the gate in the seL4 deployment described in
\S\ref{sec:architecture}). The practical ceiling on rule count is Z3 compile-time tractability
(\S\ref{sec:partitioning}), not any wire-format constraint.

\begin{table}[t]
\centering\small
\caption{Buffer layout for $E_{\mathrm{access}}$ (75 bytes).}
\label{tab:buffer-layout}
\begin{tabular}{@{}p{0.9cm}p{0.7cm}p{2.5cm}p{3.4cm}@{}}
\toprule
Offset & Bytes & Content & Description \\
\midrule
0      & 1  & \texttt{auth\_state}    & $\alpha$: role level \\
1      & 1  & \texttt{rate\_count}    & $c$: requests so far \\
2--3   & 2  & \texttt{auth\_scope}    & $\sigma$ (little-endian) \\
4--67  & 64 & 8 rules $\times$ 8 bytes & Rules $(r_1, \ldots, r_K)$ \\
68--71 & 4  & \texttt{req\_path\_id}  & $p_{\mathrm{req}}$ (little-endian) \\
72     & 1  & \texttt{req\_method}    & $m_{\mathrm{req}}$ \\
73     & 1  & \texttt{\_rate\_ok}     & Phantom (rate constraint) \\
74     & 1  & \texttt{\_access\_ok}   & Phantom (access constraint) \\
\bottomrule
\end{tabular}
\end{table}

The encoding is \textbf{total}---every valid $(\pi, q, s)$ triple
produces a buffer, so the gate never fails to decide.
It is also \textbf{injective}---distinct inputs produce distinct
buffers, so no two requests with different policy outcomes can be
confused by the validator.

\subsubsection{Structural Choices: DNF for Existential, CNF for Universal}
\label{sec:encoding-structure}

The decision function $D_{\mathrm{accept}}$ (Definition~\ref{def:policy-decision}) is
encoding-agnostic: it specifies \emph{what} the gate must decide, not
\emph{how} the decision is represented.
This section explains the structural choices made in mapping $D_{\mathrm{accept}}$ to
EverParse where-clause constraints, and why these structural choices
are required for correct fail-safe behaviour.

\paragraph{EverParse where-clauses.}

In EverParse 3D, each struct field may carry a \emph{where-clause}---a
boolean constraint that must hold for the buffer to be accepted.
For example, \texttt{UINT8 \_rate\_ok \{ rate\_count < 50 \}} declares
a \emph{phantom byte}---a field whose value is irrelevant (set to zero
by convention) but whose where-clause anchors the constraint
\texttt{rate\_count < 50}.
At compile time, Z3 verifies that the generated C parser correctly
enforces each where-clause; at runtime, the parser returns TRUE only
if all constraints hold.
The generated C code contains no conditional branches for constraint
evaluation.
In our encoding, bytes 73--74 are phantom bytes anchoring the rate
and access constraints respectively
(Table~\ref{tab:buffer-layout}).
A single phantom byte with a combined where-clause would be
functionally equivalent, but separating the rate constraint (CNF)
from the access constraint (DNF) keeps the two proof obligations
independent and preserves the fail-safe decomposition: the two
constraints have opposite default behaviours (absence of a rate
limit permits; absence of a matching rule denies).

The match function $f$ (Equation~\ref{eq:match}) decomposes naturally
into four where-clauses, and the two $\subseteq$ conditions expand to
their bitwise AND form in the encoding:

\begin{itemize}
  \item $p_i = p_{\mathrm{req}}$ --- equality over 32-bit integers
  \item $m_i \mathbin{\&} m_{\mathrm{req}} = m_{\mathrm{req}}$ ---
        $m_{\mathrm{req}} \subseteq m_i$ as bitvector
  \item $\alpha \geq \rho_i$ --- linear integer arithmetic
  \item $\sigma \mathbin{\&} \sigma_i = \sigma_i$ ---
        $\sigma_i \subseteq \sigma$ as bitvector
\end{itemize}

\paragraph{Design choices.}

The encoding uses DNF for existential constraints and CNF for
universal constraints. This choice ensures fail-safe behaviour
when policy data is incomplete.

\textbf{DNF for existential constraints.}
An existential constraint accepts if \emph{any} disjunct is true.
Under DNF, each rule slot corresponds to one disjunct; a slot whose
path is $\bot$ contributes $\FALSE$
(Corollary~\ref{cor:structural-deny}) and is therefore \emph{inert}:
it cannot grant access regardless of its remaining fields.
Policy updates replace individual rule slots at runtime; a slot that
has not yet been written contains $\bot$ by construction and defaults
to the safe state---denial---without any additional logic.

\textbf{CNF for universal constraints.}
A universal constraint requires \emph{all} conjuncts to hold.
Under CNF, an absent restriction contributes $\TRUE$, which is
safe---it does not block traffic.
The rate constraint uses this form: absence of a rate limit permits traffic.

\subsubsection{Implementation and Encoding Pipeline}

\begin{algorithm}[t]
\caption{Three-stage encoding pipeline. Every operation is a byte
         copy---no arithmetic or conditional logic.}
\label{obs:pure-copy}
\begin{algorithmic}[1]
\Require Raw HTTP request, Authenticator response, RateLimiter response
\Ensure AccessRequest buffer $\in \mathbb{B}^{75}$
\Statex
\State \textbf{Stage 1} --- $E_{\mathrm{HTTP}}$ (F*-verified):
\State \quad Raw HTTP $\to$ SecurityParamsWire (6 fields)
\Statex
\State \textbf{Stage 2} --- $E_{\mathrm{IPC}}$ (F*-verified):
\State \quad Authenticator + RateLimiter responses $\to$ SecurityParamsWire (5 fields)
\Statex
\State \textbf{Stage 3} --- $E_{\mathrm{access}}$ (unverified, auditable):
\State \quad SecurityParamsWire $\to$ AccessRequest buffer
\end{algorithmic}
\end{algorithm}

Stage~3 is a straight-line block of unconditional byte copies
implemented in
\texttt{PolicyGate/\allowbreak control\_pipeline.c}.
Direct assignments place SecurityParamsWire fields at the
AccessRequest offsets defined in Table~\ref{tab:buffer-layout}; a
\texttt{write\_rule} helper emits each rule slot's four fields
(path, method, minimum role, required scope) at 8-byte offsets
within the rule region; little-endian helpers
(\texttt{write\_u32}, \texttt{write\_u16}) handle multi-byte
fields. There are no conditionals, no dynamic allocation, and no
pointer arithmetic beyond the fixed offsets in
Table~\ref{tab:buffer-layout}. Correctness is verifiable by
inspection against the same table, and the complete source of
Stage~3 is included in the open-science artefact; its behaviour
is confirmed by 46 production tests
(\S\ref{sec:empirical-validation}).

A second encoding function (PolicyBlob, 70~bytes) handles
updates made at runtime to the policy itself; new requests are
blocked until the update completes (PolicyGate is
single-threaded). The full update flow and its verification
status are described in \S\ref{sec:architecture}.

\subsubsection{Verified Parse Function}
\label{sec:verified-parse}

\paragraph{The EverParse Toolchain.}

EverParse 3D is a verified parser generator~\cite{ramananandro_2019_everparseverifiedsecure}.
Table~\ref{tab:everparse-pipeline} shows the toolchain pipeline.

\begin{table}[t]
\centering\small
\caption{EverParse 3D compilation pipeline.}
\label{tab:everparse-pipeline}
\begin{tabular}{@{}cl@{}}
\toprule
\textbf{Stage} & \textbf{Description} \\
\midrule
\texttt{RbacPolicy.3d} & 245 lines, 3D specification \\
$\downarrow$ & EverParse compiler (3D $\to$ F*) \\
F* modules & F* type checker (verifies generated F*) \\
$\downarrow$ & Z3 SMT solver (discharges proof obligations) \\
$\downarrow$ & KreMLin extractor (F* $\to$ C) \\
\texttt{RbacPolicy.c} & 3{,}966 lines, generated C \\
\bottomrule
\end{tabular}
\end{table}

At each stage, the tool either succeeds (producing output for the next stage)
or fails with an error. If the pipeline completes without error, the generated
C code is proven correct with respect to the 3D specification
(\S\ref{sec:encoding}).

\paragraph{What Z3 Proves.}

The F* type checker translates the 3D specification into verification
conditions---logical formulas that must hold for the generated code to be
correct. These conditions are discharged by the Z3 SMT (Satisfiability Modulo
Theories) solver~\cite{demoura_2008_z3efficientsmt}: no execution of
the generated code can violate the specification.

The proof obligations discharged by Z3 can be stated as follows. Let
$V_{\mathrm{access}}$ denote the generated C function
\texttt{RbacPolicy\-Check\-Access\-Request()}. Let $C_{\mathrm{access}}$
denote the set of constraints in the \texttt{\_Access\-Request} struct of the
3D specification.

\begin{theorem}[EverParse correctness, from~\cite{ramananandro_2019_everparseverifiedsecure}]
\label{thm:everparse}
For all byte buffers $b \in \mathbb{B}^*$ and lengths $\ell \in \mathbb{N}$:
\begin{equation}
\begin{split}
V_{\mathrm{access}}(b, \ell) = \TRUE \;\;\iff\;\;
  &|b| \geq \ell \;\wedge \\
  &\text{fields of } b[0..\ell) \text{ satisfy } C_{\mathrm{access}}
\end{split}
\end{equation}
Moreover, $V_{\mathrm{access}}$:
\begin{enumerate}
  \item Is \textbf{memory-safe}: accesses only bytes within $b[0..\ell)$.
  \item Has \textbf{no integer overflow}: all arithmetic is within bounds.
  \item \textbf{Terminates}: the parser makes progress on every field.
\end{enumerate}
\end{theorem}

These properties are proven by the F* type system and Z3 together:
the type system enforces memory safety and termination through dependent
types; Z3 discharges the arithmetic and logical constraints.

The critical observation is the \textbf{universal quantification}: the
correctness theorem holds \textbf{for all} byte buffers~$b$. Z3 does not prove
correctness for specific policy values or specific requests. It proves
correctness for the validator \emph{mechanism} across all possible inputs.


\paragraph{Expanding $C_{\mathrm{access}}$.}

The constraints $C_{\mathrm{access}}$, extracted from the 3D specification, are
a formula over byte values at fixed offsets. Each per-rule predicate
$\mathrm{Match}_i(b)$ is a 4-way conjunction over byte comparisons at
fixed offsets, where
$\mathrm{u32le}(b, k)$ and $\mathrm{u16le}(b, k)$ denote little-endian
multi-byte reads:
\begin{equation}
\label{eq:match-i}
\mathrm{Match}_i(b) \;\triangleq\;
  P_i(b) \,\wedge\, M_i(b) \,\wedge\, R_i(b) \,\wedge\, S_i(b)
\end{equation}

The full constraint is:
\begin{equation}
\label{eq:c-access}
C_{\mathrm{access}}(b) \;\triangleq\;
  \underbrace{(b[1] < N_{\max})}_{\text{rate ok}} \;\wedge\;
  \underbrace{\bigvee_{i=0}^{K-1} \mathrm{Match}_i(b)}_{\text{at least one rule matches}}
\end{equation}

$\mathrm{Match}_i(b)$ is the byte-level analogue of the match function
$f(q, s, r_i)$, and $C_{\mathrm{access}}$ is the byte-level analogue
of~$D_{\mathrm{accept}}(\pi, q, s)$.

\subsection{Soundness}
\label{sec:soundness}

\subsubsection{Statement}

The main result connects the policy decision function
$D_{\mathrm{accept}}$ (Definition~\ref{def:policy-decision}) to the
verified parse function $V_{\mathrm{access}}$
(\S\ref{sec:verified-parse}) via the encoding
$E_{\mathrm{access}}$ (Definition~\ref{def:encoding}).

\begin{theorem}[Encoding soundness and completeness]
\label{thm:encoding}
For all policies $\pi$, requests $q$, and session states $s$:
\begin{equation}
V_{\mathrm{access}}\big(E_{\mathrm{access}}(\pi, q, s)\big) = \TRUE
  \quad\iff\quad D_{\mathrm{accept}}(\pi, q, s)
\end{equation}
\end{theorem}

The left-to-right direction (soundness) ensures no false acceptance:
the system never grants access it should not. The right-to-left direction
(completeness) ensures no false rejection: the system never denies access
it should not.

\subsubsection{Proof}

The proof proceeds by showing that the constraints $C_{\mathrm{access}}$, when
applied to the buffer $E_{\mathrm{access}}(\pi, q, s)$, reduce exactly to the
formula defining~$D_{\mathrm{accept}}$.

Let $b = E_{\mathrm{access}}(\pi, q, s)$. By Definition~\ref{def:encoding}
(Table~\ref{tab:buffer-layout}), each buffer position maps to a policy
language variable: $b[0] = \alpha$, $b[1] = c$,
$\mathrm{u16le}(b,2) = \sigma$, and so on for the rule and request
fields. Substituting into $C_{\mathrm{access}}(b)$:

\textbf{Rate constraint:}
\begin{equation}
b[1] < N_{\max} \;\iff\; c' < N_{\max}
\end{equation}

\textbf{Access constraint} (for each disjunct $i = 0, \ldots, K-1$):
the four conjuncts in $\mathrm{Match}_i(b)$
(Equation~\ref{eq:match-i}) reduce as follows, with the rule-level
index shifted as $k = i + 1$ to match Definition~\ref{def:policy-decision}:
\begin{align*}
\mathrm{u32le}(b, 4\!+\!8i) = \mathrm{u32le}(b, 68)
  &\;\iff\; p_k = p_{\mathrm{req}} \\
b[8\!+\!8i] \mathbin{\&} b[72] = b[72]
  &\;\iff\; m_k \mathbin{\&} m_{\mathrm{req}} = m_{\mathrm{req}} \\
b[0] \geq b[9\!+\!8i]
  &\;\iff\; \alpha \geq \rho_k \\
\mathrm{u16le}(b, 2) \mathbin{\&} \mathrm{u16le}(b, 10\!+\!8i)
  &\;\iff\; \sigma \mathbin{\&} \sigma_k = \sigma_k
\end{align*}
Together these are exactly $f(q, s, r_k)$
(Definition~\ref{def:policy-decision}). Substituting and
collecting over all $K$ rule slots gives
\[
\textstyle\bigvee_{i=0}^{K-1} \mathrm{Match}_i(b) \;=\;
\bigvee_{k=1}^{K} f(q, s, r_k).
\]

\textbf{Full substitution into $C_{\mathrm{access}}$
(Equation~\ref{eq:c-access}):}
\begin{equation}
C_{\mathrm{access}}(b) \;=\;
  (c' < N_{\max}) \;\wedge\;
  \bigvee_{k=1}^{K} f(q, s, r_k)
\end{equation}

This is identical to $D_{\mathrm{accept}}(\pi, q, s)$
(Definition~\ref{def:policy-decision}). \hfill$\square$

\subsubsection{Verification Chain}

The complete argument chains two independent proofs to conclude that
the generated C code accepts exactly what the formal model accepts:
\begin{enumerate}
  \item \textbf{EverParse/Z3} (machine-checked): the generated C
        correctly implements the 3D constraints for all byte buffers.
        \[
          V_{\mathrm{access}}(b) = \TRUE \;\iff\;
          C_{\mathrm{access}}(b)
        \]
  \item \textbf{Paper proof} (\S\ref{sec:soundness}, above):
        substituting the encoding into the constraints recovers the
        formal decision function.
        \[
          C_{\mathrm{access}}(b) \;\iff\;
          D_{\mathrm{accept}}(\pi, q, s)
        \]
\end{enumerate}
Composing (1) and (2):
\begin{equation}
V_{\mathrm{access}}\big(E_{\mathrm{access}}(\pi, q, s)\big) = \TRUE
  \;\;\iff\;\; D_{\mathrm{accept}}(\pi, q, s)
\end{equation}
The hard part---C code correctness---is machine-checked by Z3.
The paper proof reduces to substituting
Definition~\ref{def:encoding} into $C_{\mathrm{access}}$ and
verifying it recovers Definition~\ref{def:policy-decision}.
The same result holds for PolicyBlob (policy updates,
\S\ref{sec:architecture}) with the same substitution structure---a
mechanical step over the fixed encoding, mechanisable in principle
though discharged on paper here (as is
Corollary~\ref{cor:structural-deny}).
This argument depends on the trust assumptions in
Table~\ref{tab:assumptions}.

\begin{table*}[t]
\centering\small
\caption{Trust assumptions underlying the soundness argument.
A1--A6 are platform-independent (required for
Theorem~\ref{thm:encoding} on any platform that runs the generated C
code); A7--A9 apply to the seL4 deployment
(\S\ref{sec:architecture}) and bound the end-to-end guarantee---that
requests cannot bypass the gate and unverified components cannot
corrupt the verified enforcement chain---but do not affect the core
soundness result.}
\label{tab:assumptions}
\resizebox{\textwidth}{!}{%
\begin{tabular}{@{}llp{5.8cm}p{8.2cm}@{}}
\toprule
ID & Scope & Assumption & Mitigation / Status \\
\midrule
A1 & Core & Identifier separation: the sentinel \texttt{0xDEADDEAD} is not a valid path identifier
   & Enforced by F* extractor postcondition. \\
A2 & Core & Fixed-offset byte-copy glue implementing $E_{\mathrm{access}}$ (Definition~\ref{def:encoding}). Input delivery via seL4 dataports and \texttt{memcpy} is trusted per A7--A9.
   & Auditable against Table~\ref{tab:buffer-layout}; confined by seL4 isolation. Formal verification is outside scope (\S\ref{sec:related-work}). \\
A3 & Core & EverParse TCB: EverParse compiler, F* type checker, and Z3 are trusted
   & Inherited from EverParse project~\cite{ramananandro_2019_everparseverifiedsecure}; Z3~\cite{demoura_2008_z3efficientsmt} and EverParse are mature, peer-reviewed, and production-deployed~\cite{swamy_2022_hardeningattacksurfaces}. \\
A4 & Core & C compiler correctness: generated C must be compiled correctly
   & Standard assumption shared by all verification approaches except CompCert~\cite{leroy_2006_formalcertificationcompiler} or binary-level verification~\cite{klein_2009_sel4formalverification}. \\
A5 & Core & Endianness: little-endian byte order assumed for uint32/uint16 fields
   & Matches x86 target and EverParse default. \\
A6 & Core & Scope bit independence: bitwise AND correctly models capability subset only when scope bits are semantically independent
   & Dependent capabilities require explicit implication constraints. \\
\midrule
A7 & Deployment & seL4 specification correctness: Isabelle/HOL proof shows C refines abstract spec; does not guarantee spec captures intended property
   & Inherent limitation of all specification-based verification. \\
A8 & Deployment & Hardware correctness: seL4 isolation assumes correct MMU, TLB, cache coherence, and interrupt controller
   & Hardware vulnerabilities (e.g., Spectre/Meltdown class) can violate properties the kernel proof assumes. \\
A9 & Deployment & seL4 verification boundary: deployed configuration includes unverified bootstrapping code, device drivers, and platform code
   & The version deployed here is not the fully verified seL4 configuration. \\
\bottomrule
\end{tabular}}
\end{table*}

\subsubsection{Consequences}
\label{sec:consequences}

\paragraph{Separation of Mechanism and Policy.}

\begin{corollary}[Verify once, enforce all]
\label{cor:verify-once}
The generated C code $V_{\mathrm{access}}$ is compiled from the 3D
specification once and verified once by Z3. At runtime, different policies
$\pi, \pi', \pi''$ produce different byte buffers via $E_{\mathrm{access}}$,
but all are correctly enforced by the same $V_{\mathrm{access}}$.
\end{corollary}

\begin{proof}
Theorem~\ref{thm:encoding} is universally quantified over $\pi$.
\end{proof}

This is the core practical contribution, in three tiers. Traditional
verified access control re-verifies whenever the policy changes; we
verify the \emph{mechanism} once. Rule-content changes over a fixed
endpoint set are runtime data; adding endpoints or growing $K$ re-runs
the toolchain (\S\ref{sec:opnsense-coverage}) with no new hand proof;
extending the policy language requires new proofs.

\paragraph{Deny by Default as a Structural Property.}

\begin{corollary}[Structural deny by default]
\label{cor:structural-deny}
If all rules in $\pi$ are inactive ($p_i = \bot$ for all $i$),
then for all $q, s$:
\begin{equation}
\forall\, q, s:\;
  V_{\mathrm{access}}\big(E_{\mathrm{access}}(\pi, q, s)\big) = \FALSE
\end{equation}
\end{corollary}

\begin{proof}
By Definition~\ref{def:policy-decision},
\[
D_{\mathrm{accept}}(\pi, q, s) \;=\; (c' < N_{\max}) \,\wedge\,
\textstyle\bigvee_{i=1}^{K} f(q, s, r_i).
\]
When every rule is inactive, each $f(q, s, r_i)$ fails on the
path conjunct because $p_i = \bot$ cannot equal any valid
$p_{\mathrm{req}}$~(\S\ref{sec:domains}). The disjunction is
therefore $\FALSE$, so $\neg D_{\mathrm{accept}}(\pi, q, s)$ for
all $q, s$. By Theorem~\ref{thm:encoding} (contrapositive),
$V_{\mathrm{access}}$ returns $\FALSE$.
\end{proof}

This is not a runtime check that can be forgotten but a logical
consequence of the DNF structure and $\bot$ padding: no code path in
the generated validator \emph{can} grant access when all rules are
dead.

\paragraph{Isolation-Protected Encoding (seL4 Deployment).}

Theorem~\ref{thm:encoding} holds on any platform. On seL4, isolation
adds a further guarantee:

\begin{corollary}[Isolation-protected encoding]
\label{cor:isolation}
Under seL4 isolation, the only ways to influence
$b = E_{\mathrm{access}}(\pi, q, s)$ are via one of the following
paths:
\begin{enumerate}
  \item \textbf{Authenticate} and change $(\alpha, \sigma)$:
        the Authenticator (unverified, isolated) produces a
        session state $s$
        (Definition~\ref{def:session-state}) from the supplied
        credentials
  \item \textbf{Upload a new policy} and change $\pi$: via the
        policy-update path, which requires $\alpha \geq \texttt{ADMIN}$
        and is enforced by the PolicyBlob validator
        (\S\ref{sec:encoding})
  \item \textbf{Send a request} and change $(p_{\mathrm{req}}, m_{\mathrm{req}})$:
        via the F*-verified extractor (the request being
        evaluated, not trusted data)
  \item \textbf{Send requests} and increment $c$: rate counter, bounded by
        $N_{\max}$
\end{enumerate}
\end{corollary}

Paths~2--4 are mediated by verified validators or bounded counters;
Path~1 depends on the unverified Authenticator, but seL4 isolation
prevents corruption of PolicyGate's state. The encoding and proof are
platform-independent; end-to-end buffer integrity relies on the host's
isolation, which seL4 provides under Assumptions~A7--A9
(Table~\ref{tab:assumptions}).

\subsection{Platform Independence and Reference Deployment}
\label{sec:reference-deployment}
\label{sec:architecture}

Theorem~\ref{thm:encoding} holds on any platform that runs the
generated C code. Deploying the validator in production requires
two architectural guarantees outside the scope of the theorem:
every request must reach the gate, and no unverified component may
corrupt the buffer the gate reads. We demonstrate both using an
eight-component CAmkES (Component Architecture for microkernel-based
Embedded Systems) architecture on
seL4~\cite{klein_2009_sel4formalverification,
t.murray_2013_sel4generalpurpose}; a non-seL4 host (e.g., Linux
with a reverse-proxy frontend and seccomp/Landlock sandboxing)
replaces A7--A9 with its own trust base while A1--A6 are unchanged.
We have not measured a non-seL4 deployment.

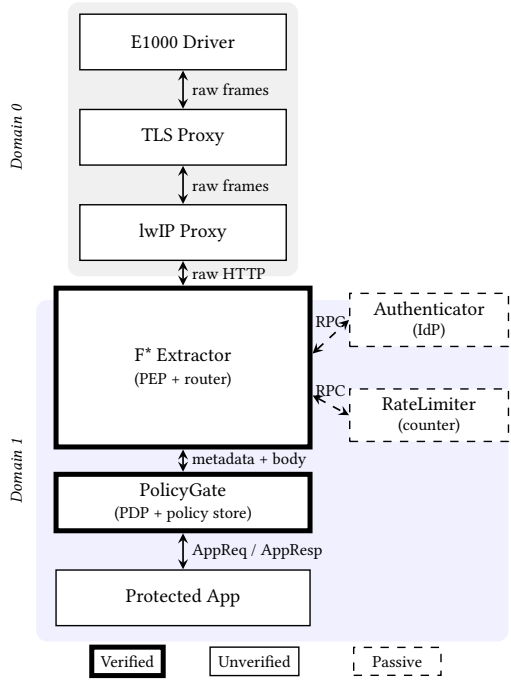
\begin{figure}[t]
\centering
\scalebox{1.05}{%
\begin{tikzpicture}[
  comp/.style={
    rectangle, draw, fill=white,
    minimum height=0.7cm,
    font=\footnotesize, align=center, inner sep=3pt
  },
  netcomp/.style={comp, line width=0.5pt, minimum width=2.6cm},
  verified/.style={comp, line width=1.8pt, minimum width=3.2cm},
  unvcomp/.style={comp, line width=0.5pt, minimum width=3.2cm},
  passive/.style={comp, line width=0.5pt, dashed, minimum width=2.0cm,
                  minimum height=0.65cm},
  pipelinearrow/.style={<->, >=stealth, semithick},
  rpcarrow/.style={<->, >=stealth, semithick, dashed}
]

\fill[gray!12, rounded corners=4pt]
  (-1.45, 0.5) rectangle (1.45, -2.95);
\fill[blue!6, rounded corners=4pt]
  (-1.85, -3.25) rectangle (4.1, -7.55);

\node[font=\scriptsize\itshape, rotate=90, anchor=center]
  at (-2.1, -1.2) {Domain 0};
\node[font=\scriptsize\itshape, rotate=90, anchor=center]
  at (-2.1, -5.4) {Domain 1};

\node[netcomp] (e1000)    at (0,  0.0)  {E1000 Driver};
\node[netcomp] (tlsproxy) at (0, -1.2)  {TLS Proxy};
\node[netcomp] (lwip)     at (0, -2.4)  {lwIP Proxy};

\node[verified, minimum height=2.0cm] (fstar) at (0, -4.1)
  {F* Extractor\\{\scriptsize(PEP + router)}};

\node[passive] (auth) at (3.1, -3.5)
  {Authenticator\\[-1pt]{\scriptsize(IdP)}};
\node[passive] (rl)   at (3.1, -4.7)
  {RateLimiter\\[-1pt]{\scriptsize(counter)}};

\node[verified] (gate) at (0, -5.8)
  {PolicyGate\\{\scriptsize(PDP + policy store)}};
\node[unvcomp]  (app)  at (0, -7.0)
  {Protected App};

\draw[pipelinearrow] (e1000.south)    -- node[right, font=\scriptsize]
  {raw frames} (tlsproxy.north);
\draw[pipelinearrow] (tlsproxy.south) -- node[right, font=\scriptsize]
  {raw frames} (lwip.north);
\draw[pipelinearrow] (lwip.south)     -- node[right, font=\scriptsize]
  {raw HTTP} (fstar.north);
\draw[pipelinearrow] (fstar.south)    -- node[right, font=\scriptsize]
  {metadata + body} (gate.north);
\draw[pipelinearrow] (gate.south)     -- node[right, font=\scriptsize]
  {AppReq / AppResp} (app.north);

\draw[rpcarrow] ([yshift=5pt]fstar.east) --
  node[above, font=\scriptsize] {RPC} (auth.west);
\draw[rpcarrow] ([yshift=-10pt]fstar.east) --
  node[above, font=\scriptsize] {RPC} (rl.west);

\node[verified, minimum width=0.9cm, minimum height=0.35cm,
      font=\scriptsize] (lv) at (-0.7, -7.8) {Verified};
\node[netcomp, minimum width=0.9cm, minimum height=0.35cm,
      font=\scriptsize] (lu) at (0.9, -7.8) {Unverified};
\node[passive, minimum width=1.1cm, minimum height=0.35cm,
      font=\scriptsize] (lp) at (2.7, -7.8) {Passive};

\end{tikzpicture}%
}
\caption{Eight-component CAmkES deployment on seL4.
Domains 0 and 1 are seL4 scheduling domains; the ingress chain
(E1000, TLS Proxy, lwIP---lightweight IP) is unverified, confined by
seL4, and terminates TLS before the F*-verified Extractor.
PEP/PDP/IdP: policy enforcement/decision point, identity provider.}
\label{fig:architecture}
\end{figure}

Figure~\ref{fig:architecture} shows the reference deployment: thick
borders mark verified components (isolated by seL4), dashed borders
mark passive unverified RPC components, and all connections are
bidirectional ring buffers. The F* Extractor parses raw HTTP,
queries the Authenticator and RateLimiter via RPC, and populates the
11 SecurityParamsWire fields. PolicyGate hosts the EverParse
validators and an in-memory policy store: it makes no RPC calls,
constructs the AccessRequest buffer from the verified metadata, and
handles policy updates via a second validator (PolicyBlob).

The F*~Extractor (8 modules, 1{,}264 lines, zero
\texttt{admit}/\texttt{assume}) is verified by
F*~\cite{swamy_2016_dependenttypesmultimonadic}/Z3 and compiled to
C by KreMLin~\cite{protzenko_2017_verifiedlowlevelprogramming} for memory safety, functional correctness, and
termination (throughout, \emph{extraction} denotes this F*-level
HTTP-to-buffer transform, distinct from the KreMLin Low*-to-C step). HTTP extraction (7~modules) produces 6~fields---path
identifier, method, bounded token, and bounded body---with proofs
against the HTTP grammar fragment; IPC extraction (1~module)
produces the remaining 5 from Authenticator and RateLimiter
responses, with role clamped ${\leq}\,2$ and subject identity
length clamped ${\leq}\,32$. All 11 SecurityParamsWire fields
(Table~\ref{tab:field-provenance}) have verified provenance. The EverParse
3D specification (\S\ref{sec:encoding}) generates two validators
from the same file: AccessRequest proves the policy decision
(Theorem~\ref{thm:encoding}), and PolicyBlob validates runtime
policy updates (admin role, rate limit, rule count). Raw HTTP bytes
are confined to the extractor's address space and destroyed after
extraction.

The Authenticator is unverified C calling
HACL*-verified~\cite{zinzindohoue_2017_haclverifiedmodern}
primitives; a bug could emit wrong $(\alpha, \sigma)$---an
authentication bug separate from enforcement correctness. seL4
isolation prevents corruption of PolicyGate's state
(Corollary~\ref{cor:isolation}); F*-verifying the Authenticator
is future work. The RateLimiter is a passive counter enforced as
a parse constraint ($c' < N_{\max}$); with PolicyGate these
components form a \emph{verified enforcement chain} between the
untrusted network and the application. Runtime policy updates are
atomic: PolicyGate is single-threaded, so the PolicyBlob
validator, rule-slot deserialisation, and struct-swap commit
execute in one critical section while incoming requests wait on
the gate's mailbox.

\subsubsection{Verified Field Provenance}

Table~\ref{tab:field-provenance} shows the 11 byte-range entries
the extractor produces per HTTP request. No unverified code writes
to this structure: 6 entries come from the F*-verified HTTP
extractor, 5 from the F*-verified IPC extractor. $\alpha$ and
$\sigma$ arrive only from IPC.Extract, structurally excluding bugs
that trust HTTP headers as authentication state
(e.g.~CVE-2024-0012~\cite{nistnationalvulnerabilitydatabase_2024_cve20240012panosauthentication});
\texttt{path\_id} is derived by byte-for-byte equality of the raw
request-target against a fixed table in the 3D spec, with no
canonicalisation in the trusted chain, so non-matching requests
deny (Corollary~\ref{cor:structural-deny}) and canonicaliser bugs
(CVE-2025-0108, CVE-2025-20362) are structurally excluded by having no
canonicaliser to confuse: distinct byte strings are distinct
identifiers, the agreement those CVEs lacked. The cost is that a
deployment needing both spellings of a resource must enumerate both,
as firmware updates. The full 3D specification
(\texttt{RbacPolicy.3d}, 245~lines) is included in the artifact.

\begin{table}[t]
\centering\small
\caption{Security parameters extracted per request (11~entries;
length and payload counted separately).
Of these, 5 (marked $\star$) are remapped into the AccessRequest
buffer (Table~\ref{tab:buffer-layout}); the remaining 6 are
forwarded to the application.}
\label{tab:field-provenance}
\begin{tabular}{@{}rrlll@{}}
\toprule
Off & Size & Field & Producer & Verified bound \\
\midrule
0     & 1  & rate\_count$\star$  & IPC.Extract   & byte copy \\
1--4  & 4  & path\_id$\star$     & HTTP.Extract  & sentinel match \\
5     & 1  & method$\star$       & HTTP.Extract  & ${\in}\{\texttt{0x00},\ldots,\texttt{0x20}\}$ \\
6     & 1  & role$\star$         & IPC.Extract   & ${\leq}\,2$ \\
7--8  & 2  & scope$\star$        & IPC.Extract   & byte copy \\
9     & 1  & sub\_id\_len & IPC.Extract   & ${\leq}\,32$ \\
10--41& 32 & subject\_id  & IPC.Extract   & bounded copy \\
42    & 1  & token\_len   & HTTP.Extract  & ${\leq}\,128$ \\
43--170& 128& token       & HTTP.Extract  & bounded copy \\
171--174& 4 & body\_len    & HTTP.Extract  & ${\leq}\,1361$ (MTU) \\
175+  & var& body         & HTTP.Extract  & bounded copy \\
\bottomrule
\end{tabular}
\end{table}

The 5 starred fields are remapped into the AccessRequest buffer
(Table~\ref{tab:buffer-layout}) by the encoding glue code
(\S\ref{sec:encoding}). The remaining 6 fields (subject identity,
token, and body) are forwarded to the application component and are
not part of the policy decision.

\section{Evaluation}
\label{sec:evaluation}

\subsection{Expressiveness and Verifiability Limits}

The policy language (Definition~\ref{def:policy}) expresses a bounded fragment
of role-based access control with capability-based scope:

\begin{itemize}
  \item Fixed $K$ rule slots ($K = 8$ in implementation)
  \item Three role levels with total order
  \item Single role per session (no multi-role activation)
  \item 16-bit scope bitfield ($2^{16}$ possible capability sets)
  \item Exact-match on path; bitmask subset-match on method (no wildcards)
  \item Global rate limit (single counter, single threshold)
\end{itemize}

Path matching uses exact 32-bit identifiers; wildcard patterns
(e.g.\ \texttt{api/module/*}) are outside the current language.
Exact-match identifiers structurally exclude the
representation-confusion class of vulnerability that motivated this
work (\S\ref{sec:introduction}). Wildcards fall outside the
recognizer's class: matching them needs variable-length, backtracking
recognition that EverParse's forward-only model cannot express.
Termination is not the obstacle---F* supports well-founded recursion,
which the extractor uses. \S\ref{sec:opnsense-coverage} validates the
boundary against the OPNsense~19.1.7 access-control surface.

The contribution is the verification methodology, not expressiveness;
the encoding generalises to more rules, more role levels, and further
universal constraints (e.g.\ body-size and content-type checks).

\textbf{Verifiability limit.} The practical limit on $K$ for monolithic
verification is not runtime performance---the validator executes in $O(K)$
time (Table~\ref{tab:wcet})---but \emph{compile-time memory consumption}
during Z3 proof discharge.
Table~\ref{tab:monolithic-scaling} shows the empirical scaling. The
scaling measurements use stand-alone AccessRequest-only 3D specs
generated at each $K$ value by the
\texttt{scaling\_experiment/} harness in the artefact; they are
distinct from the production 3D spec of the gate described in
\S\ref{sec:encoding}, which bundles the AccessRequest validator
with the policy-update validator and consequently generates more
C code at the same $K$.

\begin{table}[t]
\centering\small
\caption{Monolithic verification scaling (Z3~v4.13.3, Intel Core
Ultra~7 155H, 24\,GiB container).}
\label{tab:monolithic-scaling}
\begin{tabular}{@{}rrrrp{1.6cm}@{}}
\toprule
$K$ & Z3 Time (s) & Gen.\ C (lines) & Peak RSS (MiB) & Status \\
\midrule
4   & 10    & 850     & 1{,}010  & Verified \\
8   & 19    & 1{,}414 & 1{,}341  & Verified \\
16  & 76    & 2{,}542 & 7{,}162  & Verified \\
24  & 252   & 3{,}670 & 21{,}705 & Verified \\
32  & ${\approx}525$ & --- & 24{,}532 & OOM (SIGKILL) \\
\bottomrule
\end{tabular}
\end{table}

The $K\!=\!8 \to K\!=\!16$ transition exhibits a $5.3{\times}$ memory
increase for $2{\times}$ the rule count
(Table~\ref{tab:monolithic-scaling}), indicating a phase transition in
Z3's DPLL(T) proof search. This is a structural property of the
monolithic formula that manifests as memory exhaustion---not a time
limit. Partitioned verification avoids this entirely: each partition is
an independent 8-rule formula verified in ${\approx}16.5$\,s using
${\approx}1.2$\,GB (Table~\ref{tab:partitioned-scaling}).

\subsection{Compositional Verification via Partitioning}
\label{sec:partitioning}

The monolithic verifiability limit motivates a compositional approach.
Since the access formula is a disjunction, it can be partitioned into
independent sub-disjunctions, each verified separately.
The deployed gate is monolithic---all eight rules in one verified
check, no combiner---so this partitioned scheme is a scaling
experiment, not the deployed configuration; its OR-combiner is
\emph{unverified} glue whose only failure mode is false denial, never
false permit (\S\ref{sec:empirical-validation}).

The scheme generates $n$ EverParse 3D specifications of $K$ rules
each (DNF), plus one CNF specification for the rate limit. Each
sub-specification is verified once as a standalone artefact,
composed by glue logic.

\begin{definition}[Partitioned validators]
\label{def:partition}
Let $\pi = (r_1, \ldots, r_{nK})$ be a policy of $nK$ rules,
split into $n$ partitions of $K$ rules each: partition $j$
consists of rules $(j-1)K+1$ through $jK$, for $j = 1, \ldots, n$.
Writing $\pi_j$ for the $j$-th partition, the \emph{per-partition
match function} is
\begin{equation}
\phi_j(\pi_j, q, s) \;\triangleq\;
  \bigvee_{i=1}^{K} f(q, s, r_{(j-1)K+i}),
\end{equation}
where $f$ is the per-rule match function of
Definition~\ref{def:policy-decision}. The rate limit is enforced
separately by a 2-byte CNF specification whose where-clause is
$c' < N_{\max}$. The partitioned decision function combines the
two:
\begin{equation}
D_{\mathrm{part}}(\pi, q, s) \;\triangleq\;
  (c' < N_{\max}) \;\wedge\;
  \bigvee_{j=1}^{n} \phi_j(\pi_j, q, s).
\end{equation}
At runtime, each $\phi_j$ is implemented by an EverParse-generated
validator $V_j$ over a 73-byte AccessRequest-only DNF spec
($\alpha$, $\sigma$, the rules of $\pi_j$, $p_{\mathrm{req}}$,
$m_{\mathrm{req}}$, phantom), distinct from $V_{\mathrm{access}}$
of Theorem~\ref{thm:encoding}.
\end{definition}

\begin{lemma}[Partition equivalence]
\label{lem:partition}
For every policy $\pi$ of $nK$ rules, every request $q$, and
every session state $s$,
\[
D_{\mathrm{part}}(\pi, q, s) \;=\; D_{\mathrm{accept}}(\pi, q, s).
\]
\end{lemma}

\begin{proof}
Unfolding $\phi_j$ in $D_{\mathrm{part}}$ and reindexing
($k = (j-1)K + i$):
\begin{align*}
D_{\mathrm{part}}(\pi, q, s)
  &= (c' < N_{\max}) \wedge
     \bigvee_{j=1}^{n}\bigvee_{i=1}^{K} f(q, s, r_{(j-1)K+i}) \\
  &= (c' < N_{\max}) \wedge
     \bigvee_{k=1}^{nK} f(q, s, r_k)
  \;=\; D_{\mathrm{accept}}(\pi, q, s).
\end{align*}
Disjunction is associative and commutative, and the last step is
Definition~\ref{def:policy-decision} on the full $nK$-rule policy.
\end{proof}

\begin{table}[t]
\centering\small
\caption{Partitioned vs.\ monolithic verification scaling.}
\label{tab:partitioned-scaling}
\begin{tabular}{@{}lrrrrr@{}}
\toprule
Config & $n$ & $K$ & Total rules & Per-part.\ (s) & Total Z3 (s) \\
\midrule
Mono-8   & 1   & 8  & 8     & 19    & 19 \\
Mono-16  & 1   & 16 & 16    & 76    & 76 \\
Mono-24  & 1   & 24 & 24    & 252   & 252 \\
Mono-32  & 1   & 32 & 32    & \multicolumn{2}{c}{OOM (${>}24$\,GiB)} \\
\midrule
Part-4   & 4   & 8  & 32    & 16.4  & 66 \\
Part-16  & 16  & 8  & 128   & 16.8  & 269 \\
Part-64$^\dagger$
         & 64  & 8  & 512   & 15.6  & 999 \\
Part-256 & 256 & 8  & 2{,}048 & 15.3 & 3{,}920 \\
\bottomrule
\multicolumn{6}{@{}p{0.95\linewidth}@{}}{\footnotesize
  $^\dagger$Part-64 covers the full OPNsense~19.1.7 ACL:
  420 concrete paths fit in 512 rule slots with 92 unused slots
  filled by the inactive sentinel.} \\
\end{tabular}
\end{table}

\begin{table}[t]
\centering\small
\caption{Validator WCET (1\,M iterations, warm cache, \texttt{-O2},
Intel Core Ultra~7 155H).}
\label{tab:wcet}
\begin{tabular}{@{}lrrrrrr@{}}
\toprule
         &     &     & Total  & Buf      & Worst     & Best \\
Config   & $n$ & $K$ & rules  & (bytes)  & (ns/call) & (ns/call) \\
\midrule
Mono-4   & 1   & 4   & 4       & 43        & 12       & 8 \\
Mono-8   & 1   & 8   & 8       & 75        & 17       & 14 \\
Mono-16  & 1   & 16  & 16      & 139       & 26       & 23 \\
Mono-24  & 1   & 24  & 24      & 203       & 37       & 33 \\
\midrule
Part-4   & 4   & 8   & 32      & 294       & 61       & 15 \\
Part-16  & 16  & 8   & 128     & 1{,}170   & 261      & 16 \\
Part-64  & 64  & 8   & 512     & 4{,}674   & 992      & 15 \\
Part-256 & 256 & 8   & 2{,}048 & 18{,}690  & 3{,}944  & 15 \\
\bottomrule
\end{tabular}
\begin{flushleft}
\footnotesize Part-$n$ rows measure $n$ sequential Partition
$K\!=\!8$ calls (73 bytes each) plus the 2-byte rate constraint
$(c' < N_{\max})$; Buf column reports total bytes scanned per
worst-case decision ($73n + 2$). Dry-run measurements on host;
canonical container measurements forthcoming.
\end{flushleft}
\end{table}

\textbf{Runtime overhead.}
Buffer construction (pure \texttt{memcpy}) dominates the per-call
cost. Measured worst-case scales linearly in $n$
(Table~\ref{tab:wcet}, Part-$n$ rows): Part-64 (the
OPNsense-covering configuration) decides in 992\,ns, within 11\%
of the linear prediction. Short-circuit evaluation keeps best-case
constant at ${\sim}15$\,ns when the matching rule is in the first
partition.

\textbf{Trust analysis.}
The partitioned glue introduces no trust beyond~A2: per-partition
buffer construction uses the same fixed-offset byte copies as the
monolithic encoding (Observation~\ref{obs:pure-copy}). The verified
code scales linearly (1{,}332~lines per partition) while the unverified
glue stays constant-size, and equivalence was confirmed
experimentally: Part-2 and Mono-16 give identical decisions across all
14~test inputs (Batch~D, Group~A).

\subsection{Expressiveness: OPNsense Access Control Coverage}
\label{sec:opnsense-coverage}

We validate the expressiveness boundary against OPNsense
19.1.7~\cite{nistnationalvulnerabilitydatabase_2019_cve201911816incorrectaccess},
the real-world firewall affected by CVE-2019-11816.

\subsubsection{Coverage}

OPNsense~19.1.7 defines 201 URL patterns across 11 \texttt{ACL.xml}
files, including wildcards (e.g.\ \texttt{api/firewall/*}).
Expanding wildcards against the controller tree yields 420 unique
concrete paths. Our language covers all 420 as exact-match path
identifiers. Part-64 (512 rule slots) \emph{Z3-verifies} this
configuration in ${\sim}16.6$\,min
(Table~\ref{tab:partitioned-scaling})---a compile-time verifiability
result; end-to-end deployment uses an 8-rule subset. The CVE-2019-11816 attack
surface---login, user management, group management, firmware
update, system reboot, and status endpoints---is a subset of 7
paths at $K\!=\!7$, requiring two role levels.

\subsubsection{Wildcard Exclusion}

Our language does not support wildcard patterns. EverParse requires
fixed-size, bounded struct entries; while wildcards could in
principle be mapped to a single path identifier at the F* extractor
level, verifying that mapping is impractical. Adding new endpoints
therefore requires re-running EverParse, aligned with commercial
firmware-update practice.

\subsubsection{Relevance to CVE-2019-11816}

CVE-2019-11816 involved two enforcement bugs while the XML ACL policy
was correct. The first, in \texttt{isPageAccessible}, used
\texttt{strpos} to test whether a privilege string appeared
\emph{anywhere} in the URL, so appending \texttt{?api/core/menu}
bypassed the ACL and reached user management---enforcement logic
independent of wildcards, and exactly the class our gate eliminates
(Theorem~\ref{thm:encoding}). The second, in \texttt{urlMatch}, turned
wildcards into regexes via \texttt{str\_replace('*','.*')}, enabling
traversal through the \texttt{.*}: specific to wildcard expansion,
outside our language.

\subsection{Empirical Validation}
\label{sec:empirical-validation}

\subsubsection{Test Summary}

The formal argument is supported by 391 passing tests across five
categories (Table~\ref{tab:test-inventory}). The tests confirm that
trust assumptions not covered by the proof hold in practice, and
exercise the full pipeline including unverified components.

\begin{table}[t]
\centering
\caption{Test inventory. All tests and harnesses are included in the
artifact.}
\label{tab:test-inventory}
\small
\begin{tabular}{@{}llr@{}}
\toprule
Category & Description & Tests \\
\midrule
Methodology & \texttt{everparse\_dnf\_poc/} batches & 211 \\
Production & \texttt{pipeline\_test.c} (C1--C8) & 46 \\
F* Extractor & \texttt{verified/} (6 harnesses) & 119 \\
RateLimiter & \texttt{test\_rate\_limiter.c} & 8 \\
Integration & QEMU end-to-end (Docker) & 7 \\
\midrule
\textbf{Total} & & \textbf{391} \\
\bottomrule
\end{tabular}
\end{table}

\subsubsection{Experimental Confirmation of Formal Claims}

Table~\ref{tab:experiments} maps each formal claim to a concrete
experiment and its result. The headline confirmations are
Theorem~\ref{thm:encoding} (7/7 QEMU end-to-end) and
Lemma~\ref{lem:partition} (Part-2 vs Mono-16, 14/14 match); the
remaining rows exercise individual primitives of the encoding.

\begin{table}[t]
\centering
\caption{Experimental confirmation of formal claims.}
\label{tab:experiments}
\small
\begin{tabular}{@{}p{2.2cm}p{2.6cm}p{2.6cm}@{}}
\toprule
Claim & Experiment & Result \\
\midrule
DNF deny-default (Cor.~\ref{cor:structural-deny}) &
  C4: all-$\bot$ policy &
  ADMIN denied \\
Verify once, enforce all (Cor.~\ref{cor:verify-once}) &
  Expt.~3: same binary, different data &
  Different decisions \\
CNF overrides DNF &
  C8: rule match but rate exceeded &
  Denied (HTTP 429) \\
4-way conjunction &
  C7: all 4 conditions tested &
  8/8 correct \\
Scope subset &
  C7: operator has 0x03, rule needs 0x04 &
  Denied (correct) \\
Role hierarchy ($\alpha \!\geq\! \rho$) &
  C1: OPER $<$ ADMIN &
  Correct \\
Field provenance (11 fields) &
  IPC.Extract: 13 tests covering role, scope, subject, rate &
  13/13 pass \\
Full pipeline (Thm~\ref{thm:encoding}) &
  QEMU end-to-end &
  7/7 pass \\
\midrule
Linear compile scaling &
  Batch~D: $n\!=\!1{,}2{,}4{,}8{,}16$ at $K\!=\!8$ &
  ${\sim}16.5$\,s/part. \\
Partition equiv.\ (Lem.~\ref{lem:partition}) &
  Part-2 vs Mono-16 &
  14/14 match \\
Deny-default (partitioned) &
  All-$\bot$ across $n$ partitions &
  Deny confirmed \\
\bottomrule
\end{tabular}
\end{table}

\subsubsection{Performance}

\textbf{Runtime.} The generated validator executes in
${\sim}150$\,ns/call at \texttt{-O0}; at \texttt{-O2}, worst-case
latency is 17\,ns for $K\!=\!8$ (Table~\ref{tab:wcet}). The
validator's control flow is structurally input-independent: the
worst-case path (all $K$ disjuncts evaluated) is statically
determined by $K$ alone.

\textbf{Comparison with general-purpose engines.}
The execution times of our scheme are broadly in line with those
of other openly available access control systems, noting that the
access control models are different. The Cedar
benchmark~\cite[\S5.2]{cutler_2024_cedarnewlanguage} (EC2
m5.4xlarge) reports median latencies of ${\sim}4$--$11\,\mu$s for
Cedar, 76--$676\,\mu$s for Rego (the Open Policy Agent, OPA), and
89--$746\,\mu$s for OpenFGA (Open Fine-Grained Authorization) on
entity-graph authorization (5--50 entities with
transitive relationships); our validator measures 17\,ns at
$K\!=\!8$ and 992\,ns at Part-64 (512 rules) on flat rule
decisions. Both numbers place decision latency well below typical
network I/O.

\textbf{Compile-time verification.}
EverParse 3D-to-C compilation takes ${\sim}19$\,s at $K\!=\!8$.
The monolithic formula scales super-linearly in $K$
(Table~\ref{tab:monolithic-scaling}: OOM at $K\!=\!32$);
partitioning recovers linearity---each $K\!=\!8$ partition
verifies in ${\sim}16.5$\,s independent of $n$, so Part-256
finishes in ${\sim}65$\,min
(Table~\ref{tab:partitioned-scaling}). F* extraction verification
adds ${\sim}10.5$\,s. All are one-time costs.

\section{Discussion and Related Work}
\label{sec:related-work}

Theorem~\ref{thm:encoding} proves that the deployed C gate agrees
with the mathematical decision function for every policy, request,
and session state; it does not verify that the policy is sensible,
bound integrity past the gate, or defend against side channels.
The restricted language is deliberate: richer features such as
wildcard paths or quantified constraints would either break
EverParse's forward-only parser model or make
Z3~\cite{demoura_2008_z3efficientsmt} discharge intractable.
Assumption~A2 in Table~\ref{tab:assumptions} flags the one
hand-written component in the core chain that is not verified: a
wrong-offset bug in its fixed-offset byte-copy glue would silently
corrupt the buffer Theorem~\ref{thm:encoding} certifies. Unlike
the variable-length parsers and comparators where gateway CVEs
live (\S\ref{sec:introduction}), this glue is straight-line,
input-independent, auditable against
Table~\ref{tab:buffer-layout}, and confined by seL4 isolation.
The seL4 dataports and \texttt{memcpy} that deliver its inputs
are trusted per A7--A9, not verified here. F*-verifying Stage~3
against the $E_{\mathrm{access}}$ specification is possible but
was not undertaken: the project targeted seL4 from the outset,
so seL4's isolation already enforces the component boundary. A
deployment outside seL4 would require either an F*-verified glue
or an equivalent platform guarantee.

\subsection{Verified Parsing}

EverParse~\cite{ramananandro_2019_everparseverifiedsecure} generates
verified zero-copy parsers from binary message formats; EverParse
3D~\cite{swamy_2022_hardeningattacksurfaces} extends it with a
dependent-field DSL and refinement types. All prior uses validate
\emph{externally-defined data formats}: protocol messages
(miTLS~\cite{bhargavan_2016_mitlsverifyingprotocol}), file formats,
attested boot (DICE*~\cite{tao_2021_diceformallyverified}), and
certificates (ASN1*~\cite{ni_2023_asn1provablycorrect}, the first
formal proof of DER canonicity). PulseParse~\cite{ramananandro_2025_secureparsingserializing}
uses separation logic for non-malleable recursive formats (CBOR,
CDDL, COSE). Its copy-writer combinators could in principle
eliminate A2, but require authoring specifications directly in
Pulse with manual proof obligations---re-introducing the
verification burden that our 3D-based approach is designed to
eliminate from the policy-author's workflow. DICE* and our
work both push on EverParse's envelope but from opposite sides: Tao
et al.\ \emph{extend} LowParse with backward serialisers so ASN.1
DER's variable-length TLV can be written length-after-value, a
change that stays inside LowParse's binary-grammar model; we keep
EverParse unchanged and place a verified F* preprocessor
(\S\ref{sec:reference-deployment}) in front of it, because HTTP is a
text protocol outside any binary-grammar extension. In every case the 3D or CDDL
schema defines data \emph{structure}, and authorisation---if
present---is left to downstream code. We instead use the format
specification to encode the security policy itself: a message that
parses successfully is, by construction, an authorised request.

\subsection{Access Control Languages and Policy Verification}

Prior work verifies either specific policy instances or
unverified enforcement engines.
Hughes and Bultan~\cite{hughes_2008_automatedverificationaccess}
translate XACML (eXtensible Access Control Markup Language) policies to
Boolean formulas for SAT-based containment checks; Slaymaker et
al.~\cite{slaymaker_2010_formalisingvalidatingrbactoxacml} formalise
the RBAC-to-XACML translation in Z notation; Masi et
al.~\cite{masi_2012_formalisationimplementationxacml} give a
formal semantics for XACML with an unverified Java engine.
All three verify properties of \emph{instances} or
\emph{translations}; the policy decision point that eventually
runs the rule is trusted software. XACML~\cite{ferraiolo_2016_extensibleaccesscontrol}
and OPA~\cite{cloudnativecomputingfoundation_2021_openpolicyagent}
provide expressive languages but trust the enforcement runtime;
our approach eliminates the interpreter and replaces it with
machine-generated, machine-checked C.

\paragraph{Cedar.}
Cedar~\cite{cutler_2024_cedarnewlanguage} pursues the same goal of
verified access-control enforcement and is the closest prior work.
It verifies policy \emph{semantics} against a Lean model and connects
that model to a hand-written Rust \emph{engine} by differential random
testing. Neither approach eliminates trust; both relocate it. The
difference is structural: because we generate the enforcement code from
the specification rather than testing an engine against it, the
specification-to-code correspondence is discharged by proof, not
sampled by tests, and deny-by-default follows from the encoding
(Corollary~\ref{cor:structural-deny}) rather than being an engine
behaviour that testing must confirm. The tradeoff is real: Cedar's
richer language admits greater expressiveness with a general-purpose
engine; our narrower fragment buys a provable correspondence at the
cost of the expressiveness limits in \S\ref{sec:evaluation}.

\subsection{Verified Isolation and Composition}

seL4~\cite{klein_2009_sel4formalverification} provides the first
complete formal verification of a general-purpose OS microkernel;
Murray et al.~\cite{t.murray_2013_sel4generalpurpose} extended seL4
with information flow enforcement. We use seL4 as the isolation
substrate (\S\ref{sec:reference-deployment}). CAmkES~\cite{kuz_2007_camkescomponentmodel}
provides \emph{architectural isolation}: components interact only
through declared interfaces, but the ADL describes system topology,
not content-level access control. Our parsers fill that gap.

\section{Conclusion}
\label{sec:conclusion}

We have shown that access control enforcement can be reduced to
verified data format validation. Theorem~\ref{thm:encoding}
establishes that the generated validator accepts a request iff the
policy decision function accepts it, for every policy and every
request. A single verified binary enforces every policy
expressible in the language; policy updates are changes to runtime
bytes the proof already covers. The nine trust assumptions are
explicit in Table~\ref{tab:assumptions}: A1--A6 bound the theorem
and A7--A9 bound the seL4 deployment. The approach is supported
by a production test suite, included in the artifact
(Table~\ref{tab:test-inventory}).

\paragraph{Future work.}
The encoding methodology generalises beyond the bounded language
presented here. Near-term extensions include parameterised
constraints for wildcard path matching, additional universal CNF
clauses (body size, content type), and richer role-based features
such as partial-order role hierarchies, multi-role sessions, and
separation of duty constraints, moving toward full NIST
RBAC~\cite{ferraiolo_2001_proposedniststandard} coverage.
Longer-term, any request-response protocol with fixed-offset
security-relevant fields is amenable to the same
verify-once-enforce-all approach.

\bibliographystyle{ACM-Reference-Format}
\bibliography{bibliography}

\clearpage
\appendix

\section{Open Science}

The artifact accompanying this paper includes: the EverParse 3D
specification (\texttt{RbacPolicy.3d}, 245~lines), the F* extractor
modules (8~modules, 1{,}264~lines), the generated verified C code, the
CAmkES system configuration, the full test suite (391~tests across
5~categories as described in Table~\ref{tab:test-inventory}), build
scripts, and QEMU integration scripts for end-to-end testing.

The artifact is available at:
\begin{center}
\url{https://github.com/spanwich/http-gateway-artifact}
\end{center}

The artifact supports reproducing all results in \S\ref{sec:evaluation},
including Tables~\ref{tab:monolithic-scaling}--\ref{tab:experiments} and
the QEMU end-to-end tests. The repository contains a
\texttt{Dockerfile} and \texttt{run-tests.sh} script that build the
seL4/CAmkES image, launch QEMU, and execute all 7~integration tests
automatically.

seL4 and EverParse are open-source dependencies available from their
upstream repositories: \url{https://github.com/seL4} and
\url{https://github.com/project-everest/everparse}.

\section{Ethical Considerations}

All vulnerability analysis in this paper uses exclusively publicly
available CVE disclosures from the NIST National Vulnerability Database
(CVE-2024-0012, CVE-2025-0108, CVE-2025-20362, CVE-2019-11816). No
reverse engineering, exploitation, or vulnerability discovery was
performed. The work is entirely defensive: the contribution is a
verified enforcement mechanism, not an attack technique. The OPNsense
analysis (\S\ref{sec:opnsense-coverage}) uses the publicly available
\texttt{ACL.xml} files from the open-source OPNsense~19.1.7 release
(BSD-2-Clause).

\section{Generative AI Usage}

This paper was prepared with the assistance of Claude (Anthropic). The
tool was used in four capacities: (1)~as an interactive research tool
for exploring and refining the encoding methodology and proof structure;
(2)~for assistance in drafting and revising paper text, which the
authors reviewed, restructured, and validated against the formal
artefacts; (3)~for editorial improvements to grammar and clarity; and
(4)~for assistance in developing F* modules, the 3D specification,
test harnesses, the CAmkES ADL system configuration, and the e1000
network driver. The CAmkES ADL defines the 8-component isolation
topology that seL4 enforces (\S\ref{sec:architecture}); its correctness
is Trust Assumption~A9. The e1000 driver is unverified application code
confined by that isolation and is not part of the verified trust chain.

All AI-assisted code was validated by the same toolchain that
constitutes the paper's contribution: the F* type checker and Z3 SMT
solver independently verify the generated code for memory safety,
functional correctness, and termination. The EverParse compiler rejects
any 3D specification that fails Z3 proof discharge. The 391-test suite
(\S\ref{sec:empirical-validation}) confirms behavioural correctness of
the full pipeline. No formal claim in this paper rests on unvalidated
AI-generated content; every technical claim is either machine-checked
(Theorem~\ref{thm:encoding}) or confirmed by reproducible experiments
(Table~\ref{tab:experiments}).

All authors accept full responsibility for the accuracy, originality,
and integrity of this work.

\end{document}